\documentclass{article}

\usepackage[top=2.5cm, bottom=2.5cm, left=2.5cm, right=2.5cm]{geometry}

\usepackage{amsmath}

\usepackage[amsthm,amsmath,thmmarks]{ntheorem}

\usepackage{mathtools}

\usepackage{amssymb}
\usepackage{graphicx}
\usepackage{listings}
\usepackage[table]{xcolor}
\usepackage{diagbox}

\usepackage{tikz}
\usepackage{pgfplots}
\pgfplotsset{compat=1.9}

\usepackage{subfigure}

\usepackage{stmaryrd}
\usepackage{algorithmicx}
\usepackage[noend]{algpseudocode}
\usepackage{algorithm}

\newcommand{\cO}{\mathcal{O}}
\newcommand{\nnn}{\{1, \dots, n\}}
\newcommand\y{\cellcolor{blue!20}}
\newcommand{\zzz}{\cellcolor{red!50}}
\newcommand{\floor}[1]{\left\lfloor{{#1}}\right\rfloor}
\newcommand{\ceil}[1]{\left\lceil{{#1}}\right\rceil}
\newcommand{\ourparagraph}[1]{\smallskip\noindent\textbf{{#1}.}}

\newcommand{\lr}[2]{\llbracket {#1}; {#2} \rrbracket}
\newcommand{\comment}[1]{\textcolor{green!60!black}{/${}^*$ {#1} ${}^*$\hspace{-2pt}/}}

\DeclareMathOperator{\HH}{H}
\DeclareMathOperator{\V}{V}
\DeclareMathOperator*{\argmax}{argmax}

\newtheorem{lemma}{Lemma}
\newtheorem{theorem}{Theorem}
\newtheorem{corollary}{Corollary}
\newtheorem{conjecture}{Conjecture}
\newtheorem{definition}{Definition}
\newtheorem{assumption}{Assumption}

\begin{document}


\begin{center}
{\Large \bf Two and Three-Party Digital Goods Auctions: \\Scalable Privacy Analysis}\\

\bigskip
\bigskip
{\small Patrick Ah-Fat and Michael Huth\\
Department of Computing, Imperial College London\\
London, SW7 2AZ, United Kingdom\\
$\{$patrick.ah-fat14, m.huth$\}$@imperial.ac.uk}
\end{center}

\date{\today}

\bigskip

\begin{abstract}
A digital goods auction is a type of auction where potential buyers bid the maximal price that they are willing to pay for a certain item, which a seller can produce at a negligible cost and in unlimited quantity. 
To maximise her benefits, the aim for the seller is to find the optimal sales price, which every buyer whose bid is not lower will pay. 
For fairness and privacy purposes, buyers may be concerned about protecting the confidentiality of their bids. 
Secure Multi-Party Computation is a domain of Cryptography that would allow the seller to compute the optimal sales price while guaranteeing that the bids remain secret. 
Paradoxically, as a function of the buyers' bids, the sales price inevitably reveals some private information. 
Generic frameworks and entropy-based techniques based on Quantitative Information Flow have been developed in order to quantify and restrict those leakages. 
Due to their combinatorial nature, these techniques do not scale to large input spaces. 
In this work, we aim at scaling those privacy analyses to large input spaces in the particular case of digital goods auctions. 
We derive closed-form formulas for the posterior min-entropy of private inputs in two and three-party auctions, which enables us to effectively quantify the information leaks for arbitrarily large input spaces. 
We also provide supportive experimental evidence that enables us to formulate a conjecture that would allow us to extend our results to any number of parties. 
\end{abstract}

\section{Introduction}

Secure Multi-Party Computation (SMC) \cite{yao1982protocols,shamir1979share} is a paradigm which enables several parties to compute a public function of their own private inputs without ever disclosing their private input. 
Secure protocols that allow participants to compute such functions require them to share specific pieces of information through different rounds of communication intertwined with local computations, with the aim of guaranteeing the concealment of private values. 
Specifically, they ensure that no information flows about the private inputs, apart from that which can be inferred from the intended public output. 
From that notion of security, it follows that the output of any sensible secure computation will reveal \emph{some} information about the private inputs. 
Although cryptographic protocols have been extensively studied and optimised in the past decades in order to improve their speed and efficiency, this leakage is considered \emph{inevitable} and is commonly referred to as the \emph{acceptable leakage} in the literature, and has thus been largely ignored so far \cite{lindell2009secure,orlandi2011multiparty,%
cramer2015secure,aumann2007security}. 

We believe however, that this is a questionable position and that it is of interest~--~and of importance~--~to raise participants' awareness of this leakage before they decide to engage in an SMC protocol, and to offer them the opportunity to gauge, by themselves, the risk that they would run by entering a computation, rather than imposing this leakage on them. 
More precisely, we believe that an SMC participant may be concerned by the following questions:
is this leakage really acceptable? Is not this statement subjective? Are there objective ways of assessing the acceptability of such leakage, that each person could interpret based on her own expectations? 
Finally, is this leakage really inevitable? Under which conditions?

Some recent works have aimed at proposing some possible answers to those questions. 
A framework based on Quantitative Information Flow allows one to quantify this acceptable leakage, where the inputs' privacy is evaluated via general entropy-based measures \cite{ah2017secure,ah2018optimal,ah2018optimalArchive}. 
These measures allow participants to have an objective way of measuring the risks that taking part in a computation would present. These entropy-based measured are generic and can be parametrised so as to capture individual privacy requirements and expectations. 
Based on this model, different randomising techniques have been proposed in order to enhance participants' privacy while guaranteeing high utility \cite{ah2018optimal,ah2020protecting}. 

The principle behind those methods is to select a notion of entropy and measure the inputs' privacy via the conditional entropy of an input given the knowledge of the public output. 
Evaluating the value of those entropy measures requires to browse the whole input space and yields a complexity that is linear in the size of the total input domain. Their combinatorial essence thus does not allow those methods to scale to large inputs spaces. 
It has been shown that this complexity can be reduced for particular theoretical cases such as for three-party affine computations \cite{ah2019scalable,ah2019scalableArchive}, which allows those methods to scale to large input spaces. 
Being able to apply these methods to real-world problem requires the possibility to adapt them to real-world functions, and to large input spaces. 

The aim of this work is to focus on a particular practical application of SMC, namely the \emph{digital goods auctions}, and to scale those privacy analyses to arbitrarily large input spaces. 
More precisely, we aim at reducing the complexity of those analyses by deriving a closed-form formula for the input's posterior min-entropy in the case of two-party auctions, thus providing a way for assessing the acceptable leakage in two party auctions for any input size. 
We then notice that deriving a closed-form formula for this entropy is more involved in the presence of three parties. However, as the generic empirical methods are able to compute this entropy for small input spaces, we are mostly interested in evaluating this entropy for large input spaces. 
For three-party auctions, we thus focus on deriving an asymptotic development of this entropy for large input spaces. 
Finally, we provide supportive experimental evidence that help us to formulate a conjecture on the asymptotic behaviour of this entropy for large input spaces with any number of parties.

This paper is outlined a follows. 
We discuss some related works in Section \ref{sec:related}. We introduce relevant background in Section \ref{sec:back}. The digital goods auctions are presented in Section \ref{sec:digital}. 
Section \ref{sec:two} focuses on two-party auctions while three-party auctions are tackled in Section \ref{auction:sec:three}. 
Our conjecture is supported and formulated in Section \ref{auction:sec:multi}. 
We discuss our work in Section \ref{sec:discu} and conclude in Section \ref{sec:conclu}.



\section{Related Works}
\label{sec:related}

In this section, we present some relevant domains of cryptography and discuss their relation to our work. 

\ourparagraph{Secure Multi-party Computation}
Secure Multi-party Computation  \cite{yao1986generate,yao1982protocols,shamir1979share,%
rabin1989verifiable,ben1988completeness,chaum1988multiparty} is a domain of Cryptography that provides advanced protocols which enable several participants to compute a public function of their own private inputs without having to rely on any other trusted third party or any external authority. Those protocols enable the participants to compute a function in a decentralised manner, while ensuring that no information leaks about the private inputs, other than what can be inferred from the public output. 
The commonly called ``acceptable leakage'' which is further studied in this paper, is the information that can be inferred by an attacker about the private inputs given the knowledge of the public output alone. 

Secure Multi-Party Computation is not the only domain that is subject to an acceptable leakage. In particular, the results of our work are also applicable to other fields or scenarios that aim at protecting the inputs' privacy and that involve the opening of a public output, such as outsourced computation where a trusted third party is privately sent all the inputs and returns the public output as unique piece of information, or trusted computing where the parties input their secret data into hardware security modules, which then ensure that no unintended information will be accessible to the other parties. 

We emphasise the fact that our work focuses on the acceptable leakage that may occur in SMC, trusted computing or outsourced computations, and is thus largely orthogonal to the technicalities that SMC protocols may involve.

\ourparagraph{Differential Privacy}
Differential Privacy (DP) \cite{dwork2008differential,dwork2014algorithmic} formalises privacy concerns and introduces techniques that provide users of a database with the assurance that their personal details will not have a significant impact on the output of the queries performed on the database. More precisely, it proposes mechanisms which ensure that the outcome of the queries performed on two databases differing in at most one element will be statistically indistinguishable. Moreover, minimising the distortion of the outcome of the queries while ensuring privacy is an important trade-off that governs DP. 

Although DP is particularly suited for quantifying~--~and enhancing~--~privacy in statistical computations involving a large number of parties, its usefulness diminishes when a small number of parties are involved in the computation, or when the output of the computation is meant to be highly dependent on every input value. 
In a vote or in an auction for example, it would not be sensible to evaluate the privacy of inputs by how independent they are of the output. 
In general multi-party computations, independence between the output and the inputs is not a desirable property, and we thus need a more meaningful way of quantifying the inputs' privacy, which we discuss in the next paragraph.

\ourparagraph{Quantitative Information Flow}
The purpose of Quantitative Information Flow (QIF) \cite{smith2009foundations,malacaria2015algebraic} is to provide frameworks and techniques based on information theory and probability theory for measuring the amount of information that leaks from a secret. 
Different mathematical concepts have emerged in order to convey varied and precise information about a secret: 
Shannon entropy \cite{BLTJ:BLTJ1338} reflects the minimum number of binary questions required to recover a secret on average, while the min-entropy is an indicator of the probability to guess a secret in one try \cite{smith2011quantifying,cachin1997entropy,smith2009foundations}. 
Richer measures such as R\'enyi entropy \cite{renyi1961measures} and the $g$-entropy \cite{m2012measuring} have been introduced in order to quantify some specific properties of a secret. Generalised entropies have been proposed in order to unify those different concepts \cite{ah2018optimal,khouzani2016relative}. 

%
In this work, we will measure the information gained by an attacker by means of min-entropy, which is used extensively in Cryptography in order to quantify the vulnerability of a secret. 
Although we selected the conditional min-entropy in order to propose a measure of privacy that can be meaningful in SMC, we believe that it would be interesting to extend and compare our approach to other notions of entropy and possibly other methods for quantifying privacy. 

\section{Background}
\label{sec:back}

Let us now present the mathematical model \cite{ah2017secure,ah2018optimal} that we will use to study the notion of privacy in SMC. 

Let $n$ be a positive integer. Let $x_1, \dots, x_n$ be $n$ integers, belonging to $n$ different parties $P_1, \dots, P_n$ respectively. 
Let us assume that these parties wish to enter the secure computation of an $n$-ary function $f$ and to compute its output $o = f(x_1, \dots, x_n)$. 

We are interested in studying the information that opening the output reveals about private inputs. 
More precisely, let us assume that we wish to study the information that leaks about private input $x_j$. We call it \emph{targeted} input, while the other inputs are called \emph{spectators}' inputs. 
To this end, we consider each input $x_i$ as a random variable $X_i$ taking values in a domain $D_i$. 
The output is also assigned a random variable $O$ defined as a composition of random variables $O = f(X_1, \dots, X_n)$. Its domain is denoted by $D_O$. 

Then the privacy of targeted input $X_j$ will be quantified as the conditional min-entropy of $X_j$ given $O$, defined as:
\[ \HH(X_j \mid O) = - \log \V(X_j \mid O) \]
where the conditional vulnerability $\V(X_j \mid O)$ is defined as:
\begin{equation}
\label{eq:vuln_gen}
\V(X_j \mid O) = \sum_{o \in D_O} p(O=o) \cdot \max_{x_j \in D_j} p(X_j=x_j\mid O=o)
\end{equation}

For clarity purposes, we will abuse notation and omit the domains in the summations, and omit the random variable name in the probability notations, when they can be obviously inferred from context. This way, the above vulnerability rewrites as 
\( \V(X_j \mid O) = \sum_{o} p(o) \cdot \max_{x_j} p(x_j\mid o) \)
. 

We now formulate an assumption that will hold throughout the paper. 

\begin{assumption}
\label{ass:uni}
Throughout the paper, we assume that the inputs are uniformly distributed over $\lr{1}{m}$ where $m$ is a positive integer, and where $\lr{1}{m}$ denotes $\{1, \dots, m\}$. 
\end{assumption}


By virtue of Bayes' theorem and Assumption \ref{ass:uni}, the vulnerability from Equation (\ref{eq:vuln_gen}) can be rewritten as:
\begin{align}
\V 
&= \sum_o \max_x p(x) \cdot p(o \mid x) \nonumber\\
&= \frac{1}{m} \sum_o \max_x p(o \mid x) \label{auction:eq:vuln}
\end{align}

We recall that all the values of $p(o \mid x)$ can be computed in $\cO(m^n)$, the $\max$ can be computed in $\cO(m)$ and the sum in $\cO(m^n)$, which quickly becomes intractable as the input size $m$ grows. 
We thus seek a closed-form formula for $\HH(X \mid O)$. 
In the next section, we present the function $f$ that is considered in this paper. 
Sections \ref{sec:two} and \ref{auction:sec:three} then focus on simplifying the expression of $\V$ for this precise function $f$. 
More precisely, we derive a closed-form formula for $\V$ for two-party auctions. 
Following the same approach, we notice that deriving an \emph{exact} formula would be more involved in the three-party case. However, our main objective is to be able to provide analyses that scale to \emph{large} input spaces, since empirical, combinatorial, analyses are already able to compute exact values of $\V$ for \emph{small} input spaces~--~and fail to do so for large ones. 
We then decide to focus on deriving the \emph{asymptotic} behaviour of $\V$ for large values of $m$ in three-party auctions.

\section{Digital Goods Auctions}
\label{sec:digital}

Auctions are part of the practical use cases that can benefit from the security properties provided by Secure Multi-Party Computation. In fact, and as an aside, one of the first practical applications of SMC implemented on a large scale was an auction between several Danish farmers and a producer \cite{bogetoft2009secure}. 
Indeed, depending on the setting and the rules of the auction, participants may be interested in keeping their bids private in order to protect their economic interests. Resorting to SMC might also enhance fairness between participants and may provide other reassuring guarantees that may be lacking in a traditional auction. A non-exhaustive list of such guarantees are:
\begin{itemize}
\item The confidentiality of the bids protects the participants' economic position from both the auctioneer and the other participants. 
\item Protocols that are secure under active adversaries may guarantee the participants that the result has not been falsified. In comparison with traditional methods, this prevents the polling authority from being involved in any kind of corruption. 
\item In some cases, SMC may offer the benefit that bids from different participants are taken into account simultaneously, whereas some traditional auction types may not. 
\end{itemize}

Although SMC provides the participants with a way of entering all their inputs once and simultaneously, and importantly, without revealing their bids, we know that some information will leak about private bids. 

In this work, we study a particular case of auctions, known as \emph{digital goods auctions}. 
This application has also been chosen as a case study in influential papers on Differential Privacy such as McSherry and Talwar's paper on the exponential mechanism \cite{mcsherry2007mechanism}. 
Let us introduce the principle of a digital goods auction, and explain the different pieces of private information that are being manipulated and the public information that is revealed during such auctions. 

A digital goods auction involves one seller and $n$ buyers. 
A seller has an unlimited supply of a certain item or good, that she wants to sell. 
Each buyer will either buy the item, or refuse to buy it. In particular, a buyer will not buy the item several times. 
Each buyer $P_i$ will bid a price $x_i$, which is the maximal price that he is willing to pay to buy the item. 
If the sales price $p$ of the item is greater than $x_i$, then buyer $P_i$ will not buy the item. If $p$ is not greater than $x_i$, then $P_i$ will pay the price $p$ to get the item, which will turn into benefits for the seller. 
We assume that the seller did not pay anything to acquire the items, so that her total \emph{benefits}~--~also referred to as \emph{budget} or \emph{profit}~--~will equal $pb$ where $p$ is the sales price of the item and $b$ is the number of buyers who can afford it, i.e. $b = |\{ i \in \lr{1}{n} \mid x_i \geq p \}|$. 
The aim of the auction is to determine the optimal sales price of the item that maximises the seller's benefits. As an aside, if the seller's profits can be maximised with different values of $p$, then we define the auction as retaining the lowest value of $p$, as it will satisfy more participants. 
%
%
%
The computation of the optimal price of the item can be represented as the following function $f$ described in Algorithm \ref{auction:algo:f}:

{\small
\begin{algorithm}
\hspace*{\algorithmicindent} \textbf{Inputs:} $x_1, \dots, x_n  \in \lr{1}{m}$ \\
\hspace*{\algorithmicindent} \textbf{Output:} Auction sales price $p \in \{x_1, \dots, x_n\}$
\vspace{1.5mm}
\begin{algorithmic}[1]
\Function{$f$}{$x_1, \dots, x_n$}
\State sort $x_i$'s in descending order such that $x_1 \geq \dots \geq x_n$
\State \( k \gets \argmax_j j \cdot x_j \)   (choose largest possible $k$)
\State \Return $x_k$
\EndFunction
\end{algorithmic}
\caption{Multi-party auction function $f$}
\label{auction:algo:f}
\end{algorithm}
}

%

Choosing the largest possible $k$ means that if the same budget is attainable with different sales prices, we choose the one which enables more participants to buy, e.g. $f(1, 1, 4, 1) = 1$ where the maximal budget can equally be achieved with a sales price of $1$ or $4$, but the former allows $4$ participants to buy whereas the latter allows only one. 
It is worth noting that the output $o = f(x_1, \dots, x_n)$ necessarily equals one of the input values $x_i$. We can show that the price is not optimal otherwise. 

Naturally, the buyers' bids $x_i$ constitute private pieces of information that the buyers do not wish to reveal: neither the other buyers nor the seller should be able to learn any information about a particular bid $x_i$ before the opening of the final price $p$. 
In order to guarantee such a level of privacy, the participants can for example enter an SMC protocol, or resort to a Trusted Execution Environment. 

On the other hand, the sales price $p$ is the information that is intended to be computed and to be made public. 
As such, it inevitably reveals some information about the private bids, which is commonly referred to as the acceptable leakage in the SMC literature. 
One may wish to gauge this acceptable leakage, and in particular may wonder whether this leakage is tolerable in a two-party auction. 

The aim of the work reported in this paper is to quantify the information that flows about the private bids when the result of the auction~--~i.e. the sales price~--~is revealed. 
More precisely, we aim at deriving a method for quantifying the bids' privacy that is \emph{scalable} to arbitrarily large input spaces, 
which previous generic methods were not able to accommodate. 
Precisely, we now aim at simplifying the expression of $p(o \mid x)$ from Equation (\ref{auction:eq:vuln}), where $x$ represents one targeted input in order to be able to compute it for large values of $m$.

\section{Two-Party Auctions}
\label{sec:two}

In a two-party auction, the function $f$ can be simplified. 
It is straightforward to see that the sorting procedure and the $\argmax$ function can be written as in the following Algorithm \ref{auction:algo:f2}. 

{\small
\begin{algorithm}
\hspace*{\algorithmicindent} \textbf{Inputs:} $x, y \in \lr{1}{m}$ \\
\hspace*{\algorithmicindent} \textbf{Output:} Sales price $p \in \{x, y\}$
\vspace{1.5mm}
\begin{algorithmic}[1]
\Function{$f$}{$x, y$}
\If {$x > y$}
{        \If {$x > 2y$} 
        	 \Return $x$
        \Else 
             \phantom{\textbf{if} $x > 2y$ \textbf{t}}
             \Return $y$
		\EndIf
}
\Else ~\comment{$x \leq y$}
    	\If {$y > 2x$}  
    	      \Return $y$
    	\Else  
             \phantom{\textbf{if} $x > 2y$ \textbf{t}}
    		 \Return $x$
    	\EndIf
\EndIf
\EndFunction
%
\end{algorithmic}
\caption{Two-party auction function $f$}
\label{auction:algo:f2}
\end{algorithm}
} 


In order to illustrate and reason about the results of such a function, let us assume that $m=9$ and let us plot the function's outputs on the 2-dimensional array in Table \ref{auction:table:foutputs}.

\begin{table}
\caption{Outputs of the two-party auction with maximal input $m=9$. }
\label{auction:table:foutputs}
\[
  \begin{array}{c|ccccccccc}
    $\diagbox[width=2.5em,height=2.5em]{$x$}{$y$}$  & 1  & 2 & 3&4&5&6&7&8&9 \\ \hline
    1   & 1 & 1 & 3 & 4 & 5 & 6 & 7 & 8 & 9 \\
    2   & 1 & 2 & 2 & 2 & 5 & 6 & 7 & 8 & 9 \\
    3   & 3 & 2 & 3 & 3 & 3 & 3 & 7 & 8 & 9 \\
    4   & 4 & 2 & 3 & 4 & 4 & 4 & 4 & 4 & 9 \\
    5   & 5 & 5 & 3 & 4 & 5 & 5 & 5 & 5 & 5 \\
    6   & 6 & 6 & 3 & 4 & 5 & 6 & 6 & 6 & 6 \\
    7   & 7 & 7 & 7 & 4 & 5 & 6 & 7 & 7 & 7 \\
    8   & 8 & 8 & 8 & 4 & 5 & 6 & 7 & 8 & 8 \\
    9   & 9 & 9 & 9 & 9 & 5 & 6 & 7 & 8 & 9 \\
  \end{array}
\]
\end{table}

Let us now look at the quantity that we wish to compute. 
In order to compute $\HH(X\mid O)$, we will compute $\max_x p(o\mid x)$ for each output $o$. 
To do so, we argue that for each output $o$, we have $\max_x p(o\mid x) = p(O=o \mid X=o)$. 
Indeed, for fixed values of $o$ and $x$, we have:
\begin{align*}
 p(o \mid x) &= \sum_{\substack{y \\ f(x, y) = o}} p(y) \\
 &= \sum_{\substack{y \\ f(x, y) = o}} \frac{1}{m} 
\end{align*}
since the inputs are uniformly distributed. 

Moreover, if $x \neq o$, then this sum can contain \emph{at most} one summand which would correspond to the case where $y=o$, since the output $o$ must equal one of the inputs. 
On the contrary, if $x=o$ then the sum contains \emph{at least} one summand which corresponds to the case where $y=o$, again because the output must equal one of the inputs. 

For that reason, we have:
\[ \max_x  p(o \mid x) = p(o\mid X=o) \]
and thus:
\begin{align}
\V &= \frac{1}{m} \sum_o p(o\mid X = o) \nonumber \\
&= \frac{1}{m^2} \sum_o |\{ y \mid f(o, y) = o \}| \label{eq:vuln_two}
\end{align}
where $|\cdot|$ denotes the cardinality of a set.

We can now illustrate in Table \ref{auction:table:colour} the result of that sum by highlighting in colour all the cells that satisfy the condition $f(o, y) = o$. 
We gather all the inputs satisfying this condition in a set that we define as $S = \{ (x, y) \in \lr{1}{m}^2 \mid f(x, y) = x \}$. 
Red cells correspond to the cases where $x>y$ while blue cells to those where $x \leq y$:

\begin{table}
\caption{Enumerating set $S$ that contains all the input combinations $(x,y)$ that satisfy $f(x, y) = x$. }
\centering
\vspace{.4\baselineskip}
\subfigure[Couloured cells highlight elements of $S$. Red cells correspond to cases where $x > y$. Blue cells include cases where $x \leq y$. ]{\label{auction:table:colour}
$
  \begin{array}{c|ccccccccc}
    \multicolumn{10}{c}{} \vspace{-2mm} \\
    $\diagbox[width=2.5em,height=2.5em]{$x$}{$y$}$  & 1  & 2 & 3&4&5&6&7&8&9 \\ \hline
    1   & \y 1 & \y 1 & 3 & 4 & 5 & 6 & 7 & 8 & 9 \\
    2   & 1 & \y 2 & \y 2 & \y 2 & 5 & 6 & 7 & 8 & 9 \\
    3   & \zzz 3 & 2 & \y 3 & \y 3 & \y 3 & \y 3 & 7 & 8 & 9 \\
    4   & \zzz 4 & 2 & 3 & \y 4 & \y 4 & \y 4 & \y 4 & \y 4 & 9 \\
    5   & \zzz 5 & \zzz 5 & 3 & 4 & \y 5 & \y 5 & \y 5 & \y 5 & \y 5 \\
    6   & \zzz 6 & \zzz 6 & 3 & 4 & 5 & \y 6 & \y 6 & \y 6 & \y 6 \\
    7   & \zzz 7 & \zzz 7 & \zzz 7 & 4 & 5 & 6 & \y 7 & \y 7 & \y 7 \\
    8   & \zzz 8 & \zzz 8 & \zzz 8 & 4 & 5 & 6 & 7 & \y 8 & \y 8 \\
    9   & \zzz 9 & \zzz 9 & \zzz 9 & \zzz 9 & 5 & 6 & 7 & 8 & \y 9 \\
  \end{array}
$
}
\hspace{1.5cm}
\subfigure[Blue cells have been transposed from Table \ref{auction:table:colour} in order to surface an obvious expression for the cardinality $|S|$. ]{\label{auction:table:tranpose}
$
\begin{array}{c|ccccccccc}
    \multicolumn{10}{c}{} \vspace{-2mm} \\
    $\diagbox[width=2.5em,height=2.5em]{$x$}{$y$}$  & 1  & 2 & 3&4&5&6&7&8&9 \\ \hline
    1   & \y 1 & 1 & 3 & 4 & 5 & 6 & 7 & 8 & 9 \\
    2   & \y 1 & \y 2 & 2 & 2 & 5 & 6 & 7 & 8 & 9 \\
    3   & \zzz 3 & \y 2 & \y 3 & 3 & 3 & 3 & 7 & 8 & 9 \\
    4   & \zzz 4 & \y 2 & \y 3 & \y 4 & 4 & 4 & 4 & 4 & 9 \\
    5   & \zzz 5 & \zzz 5 & \y 3 & \y 4 & \y 5 & 5 & 5 & 5 & 5 \\
    6   & \zzz 6 & \zzz 6 & \y 3 & \y 4 & \y 5 & \y 6 & 6 & 6 & 6 \\
    7   & \zzz 7 & \zzz 7 & \zzz 7 & \y 4 & \y 5 & \y 6 & \y 7 & 7 & 7 \\
    8   & \zzz 8 & \zzz 8 & \zzz 8 & \y 4 & \y 5 & \y 6 & \y 7 & \y 8 & 8 \\
    9   & \zzz 9 & \zzz 9 & \zzz 9 & \zzz 9 & \y 5 & \y 6 & \y 7 & \y 8 & \y 9 \\
  \end{array}
$}
\end{table}

The aim is thus to compute the number of coloured cells, which equals the desired values of $\sum_o |\{ y \mid f(o, y) = o \}|$. 
To do this, we present the following geometric interpretation. Let us ``transpose'' all the cells highlighted in blue. By transposing a cell at location $(i, j)$, we mean discolouring this cell and then colouring its symmetric one at location $(j, i)$. 
We obtain the array depicted in Table \ref{auction:table:tranpose}.

We note that no two cells are coloured with two different colours. 
From that we can see that:
\begin{align*}
\sum_o |\{ y \mid f(o, y) = o \}| &= \sum_{k=1}^m k
\end{align*}
which leads us to the following result. 

\begin{theorem}
\label{auction:thm:h2}
In a two-party auction, where the inputs are uniformly distributed over $\lr{1}{m}$, we have:
\begin{align*}
\HH(X \mid O) 
&= - \log \frac{m+1}{2m}
\end{align*}
\end{theorem}

\begin{proof}

As we argued in Equation (\ref{eq:vuln_two}), we know that:
\begin{align*}
\V 
&= \frac{1}{m^2} \sum_o |\{ y \mid f(o, y) = o \}|
\end{align*}

It thus suffices to compute the cardinality of the following set $S$:
\[ S = \{ (x, y) \in \lr{1}{m}^2 \mid f(x, y) = x \} \]

Let us express $S$ as the disjoint union of the following two subsets:
\begin{align*}
S_1 &= \{ (x, y) \in \lr{1}{m}^2 \mid f(x, y) = x \wedge x > y \} \\
S_2 &= \{ (x, y) \in \lr{1}{m}^2 \mid f(x, y) = x \wedge x \leq y \}
\end{align*}

We note that $S_1$ corresponds to the area highlighted in red in Table \ref{auction:table:colour}, and $S_2$ corresponds to the blue area. 
The additional ordering on $x$ and $y$ ensures that $\{S_1, S_2\}$ forms a partition of $S$ and thus $|S| = |S_1|+|S_2|$. 

We can swap both coordinates of the elements of $S_2$ without altering its cardinality. We thus have $|S_2| = |S_2'|$ where we define $S_2'$ as follows:
\begin{align*}
S_2' &= \{ (x, y) \in \lr{1}{m}^2 \mid f(y, x) = y \wedge y \leq x \} \\
&= \{ (x, y) \in \lr{1}{m}^2 \mid 2y \geq x \wedge y \leq x \}
\end{align*}

On the other hand, we have:
\[ S_1 = \{ (x, y) \in \lr{1}{m}^2 \mid x > 2y \wedge y \leq x \} \\ \]

We can see that $S_1$ and $S_2'$ are disjoint and thus $|S_1| + |S_2'| = |S_1 \cup S_2'|$. 
This union can be rewritten as:
\begin{align*}
S_1 \cup S_2' &= \{ (x, y) \in \lr{1}{m}^2 \mid (x > 2y \vee 2y \geq x) \wedge y \leq x \} \\
&= \{ (x, y) \in \lr{1}{m}^2 \mid y \leq x \}
\end{align*}
from which we can infer that
\begin{align*}
|S_1 \cup S_2'|
 &= \frac{m(m+1)}{2} 
\end{align*}

and thus:
\begin{align*}
\HH(X \mid O) &= - \log \frac{m(m+1)}{2m^2} \\
&= - \log \frac{m+1}{2m}
\end{align*}
\end{proof}

We also checked that this formula is experimentally validated by our programs computing $\HH(X \mid O)$ empirically. 

As we are particularly interested in studying the behaviour of $\HH(X \mid O)$ for large input spaces, we formulate the following corollary. 

\begin{corollary}
\label{auction:coro:2}
When $m$ tends towards infinity, $\HH(X \mid O)$ converges and:
\[ \lim_{m \to \infty} \HH(X \mid O) = \log 2 \]
\end{corollary}

\begin{proof}
This is an immediate consequence of Theorem \ref{auction:thm:h2}. 
\end{proof}

\section{Three-Party Auctions}
\label{auction:sec:three}

Let us now consider the case where three parties enter an auction. Algorithm $f$ is adapted in the following Algorithm \ref{auction:algo:f3}. 
We can see that once the inputs have been sorted, the program simply consists in finding the maximum of $x, 2y$ and $3z$.

{\small
\begin{algorithm}
\hspace*{\algorithmicindent} \textbf{Inputs:} $x, y, z \in \lr{1}{m}$ \\
\hspace*{\algorithmicindent} \textbf{Output:} Sales price $p \in \{x, y, z\}$
\vspace{1.5mm}
\begin{algorithmic}[1]
\Function{$f$}{$x, y, z$}
\State Sort inputs such that $x \geq y \geq z$
\If {$x > 2y$ and $x > 3z$} \Return $x$
\Else
    	\If {$2y > 3z$} \Return $y$
    	\Else \phantom{\textbf{if} $2y > 3z$ \textbf{t}} \Return $z$
    	\EndIf
\EndIf
\EndFunction
\end{algorithmic}
\caption{Three-party auction function $f$}
\label{auction:algo:f3}
\end{algorithm}
} 


Before computing $\V$, let us introduce some useful results that we will need in this section. 

\subsection{Preliminaries}

Let us generalise a result that we hinted at in the previous section, and that will be of importance in deriving an algebraic expression for $\HH(X \mid O)$. 
We claim that substituting the value of one of the inputs to the value of the output does not change the output of $f$. 

\begin{lemma}
\label{auction:lemma:replace}
Let $\{x_i\}_i$ be a set of $n$ inputs and let $o$ be a value. Then we have:
\begin{equation}
f(x_1, x_2, \dots, x_n) = o \implies f(o, x_2, \dots, x_n) = o
\end{equation}
\end{lemma}

\begin{proof}
Let us introduce a few notations in order to develop this proof. 
Let $\{x_i\}_i$ be a set of $n$ inputs and let $o$ and $o'$ two values such that:
\begin{align}
f(x_1, \dots, x_n) &= o \label{auction:eq:comp1} \\
f(o, x_2, \dots, x_n) &= o' \label{auction:eq:comp2}
\end{align}
where we will refer to Equation (\ref{auction:eq:comp1}) to scenario 1 while Equation (\ref{auction:eq:comp2}) will be referred to as scenario 2. 

For sake of convenience, let us introduce another set of inputs $\{x_i'\}_i$ which we define as:
\begin{align}
x_1' &= o \\
\forall i \in \lr{2}{n} \colon x_i' &= x_i
\end{align}
so that we can write $f(x_1', \dots, x_n') = o'$. 

Let $S$ and $S'$ be two subsets of $\lr{1}{n}$ defined as:
\begin{align*}
S &= \{ i \in \lr{1}{n} \mid x_i \geq o \} \\
S' &= \{ i \in \lr{1}{n} \mid x_i' \geq o' \}
\end{align*}

In other words, $S$ and $S'$ represent the sets of parties that will get to buy the item in an auction involving the sets of inputs $\{x_i\}_i$ and $\{x_i'\}_i$ respectively.

Let $b$ and $b'$ be the seller's benefits in both scenarios, defined as:
\begin{align*}
b &= |S|\cdot o \\
b' &= |S'| \cdot o'
\end{align*}

Finally, for input set $\{\xi_i\}_{i}$, subset $\sigma \subseteq \lr{1}{n}$ and value $\omega$, we say that the triple $(\{\xi_i\}_{i}, \sigma, \omega)$ is \emph{qualified} if:
\[ \forall j \in \sigma \colon \xi_j \geq \omega \]

We can notice that if $(\{\xi_i\}_{i}, \sigma, \omega)$ is qualified, then the seller's benefits in an auction involving inputs $\{\xi_i\}_{i}$ will be no lower than $|\sigma| \cdot \omega$.

The intuition of the proof follows. 
First, there is nothing to show when $o=x_1$. 
We will prove that if $o<x_1$, then $b'=b$ as $o$ is the optimal price for scenario 2. 
If $o>x_1$, we will argue that $b' = b+o$ which can only be achieved for $o'=o$. 

\ourparagraph{Case 1}
Let us assume that $o < x_1$. 
We first argue that $b'=b$. 

In this case, we have $x_i' \leq x_i$ for all $i$ in $\lr{1}{n}$. 
For all subset $\sigma \subseteq \nnn$, if $(\{x_i'\}, \sigma, o')$ is qualified, then $(\{x_i\}, \sigma, o')$ is also qualified. Thus we have $b' \leq b$. 
Moreover, we know that $(\{x_i'\}, S, o)$ is qualified since $x_1'\geq o$. 
Hence, $b' \geq b$ and thus $b'=b$. 

Let us now argue that $o'=o$. 
We know that $b'$ can be achieved with a sales price of $o$. As the auction function $f$ favours lower sales prices for the same seller's benefits, we know that $o' \leq o$. 

Let us assume by contradiction that $o' < o$. 
As argued before, we know that $(\{x_i\}, S', o')$ is qualified. 
The benefits with inputs $\{x_i\}$ would equal $|S'|\cdot o'$. But we have just shown that $b' = b$ where by definition $b'$ equals $|S' \cdot o'|$. 
Thus the optimal benefits in scenario 1 would also be achieved with an output value $o'$ satisfying $o' < o$ which is a contradiction. 
Thus $o'=o$. 

\ourparagraph{Case 2}
Let us assume that $o > x_1$. 
Let us first argue that $b'=b+o$. 

We know that $(\{x_i\}, S, o)$ is qualified by definition and thus $(\{x_i'\}, S \cup \{1\}, o)$ is also qualified since $x_1'=o$. 
Thus $b' \geq (|S| + 1) \cdot o$, or in other words $b' \geq b + o$. 

Moreover, we know that for all subset $\sigma \subseteq \nnn$, if $(\{x_i'\}, \sigma, o')$ is qualified, then the triple $(\{x_i\}, \sigma \cap \{2, \dots, n\}, o')$ is also qualified, since $x_i' = x_i$ for all $i$ in $\lr{2}{n}$. 
Let us now assume by contradiction that $b' > b + o$. 
As by definition $(\{x_i'\}, S', o')$ is qualified, we know that $(\{x_i\}, S' \cap \lr{2}{n}, o')$ is qualified, too. 
Thus: 
\begin{align*}
b &\geq (|S'|-1) \cdot o' \\
&\geq |S'| \cdot o' - o' \\
&\geq b' - o'
\end{align*}

By assumption, this implies:
\begin{align*}
b > b + o - o'
\end{align*}
and thus:
\begin{align*}
o' > o
\end{align*}

This implies that $1 \notin S'$ and thus $(\{x_i\}, S', o')$ is qualified, which means that $b \geq b'$, which is a contradiction. 

In conclusion, we have $b' = b+o$. Let us now argue that $o'=o$. 

We have already mentioned that $b'$ can be achieved with a sales price of $o$, and thus $o' \leq o$. 
Let us assume by contradiction that $o' < o$. 
Then we know that $1 \in S'$ (since $x_1'=o$). 
We also know that as $(\{x_i'\}, S', o')$ is qualified, then $(\{x_i\}, S \setminus \{1\}, o')$ is qualified too. 
Consequently:
\begin{align*}
b &\geq (|S'| - 1) \cdot o' \\
&\geq b' - o'
\end{align*}

As assumed by proof by contradiction that $o' < o$, we thus have $b > b' - o$ and thus:
\[ b' < b+o \]
which is a contradiction, which concludes the case and the proof. 
%
%
%
\end{proof}

We also recall a result that will be useful for studying the asymptotic behaviour of $\HH(X \mid O)$. 

\begin{lemma}
\label{auction:lemma:sums}
Let $n$ be a positive integer and let $a$ and $b$ be positive integer no larger than $n$. 
Then:
\begin{align*}
\sum_{k=a}^b k
&= \frac{1}{2}(b^2 - a^2) + \cO(n) \\
\sum_{k=a}^b k^2
&= \frac{1}{3}(b^3 - a^3) + \cO(n^2)
\end{align*}
\end{lemma}

\begin{proof}
We recall that we have:
\begin{align*}
\sum_{k=1}^n k &= \frac{n(n+1)}{2} \\
\sum_{k=1}^n k^2 &= \frac{n(n+1)(2n+1)}{6}
\end{align*}

Thus for all positive integer $c$ no larger than $n$, we have:
\begin{align*}
\sum_{k=1}^c k &= \frac{1}{2}c^2 + \cO(n) \\
\sum_{k=1}^c k^2 &= \frac{1}{3}c^3 + \cO(n^2)
\end{align*}

And thus as $a$ and $b$ are no larger than $n$:
\begin{align*}
\sum_{k=a}^b k &= \sum_{k=1}^b k - \sum_{k=1}^a k + \cO(n) \\
&= \frac{1}{2}(b^2 - a^2) + \cO(n) \\
\sum_{k=1}^n k^2 &= \sum_{k=1}^b k^2 - \sum_{k=1}^a k^2 + \cO(n^2) \\
&= \frac{1}{3}(b^3 - a^3) + \cO(n^2)
\end{align*}
where we note that the bounds of the indices in the sums are allowed to differ by 1 since the difference is compensated in the $\cO(n)$ and $\cO(n^2)$ terms. 
\end{proof}

\subsection{Deriving an Expression for the Input Vulnerability}

We are now interested in computing the value of $\V$. 
For this, we argue again that $\max_x p(o\mid x) = p(o \mid X = o)$. 
Indeed, we know that for fixed values of $o$ and $x$, we have:
\begin{align*}
p(o\mid x)&= \sum_{\substack{y, z \\ f(x, y, z) = o}} p(y, z) \\
&= \frac{|\{(y, z) \mid f(x, y, z) = o\}|}{m^2}
\end{align*}
where we note that in this 3-party setting, the pair $(y,z)$ plays the role of the spectators' input, which comprises two inputs $y$ and $z$. 
However, Lemma \ref{auction:lemma:replace} ensures that:
\[ \forall x \colon |\{(y, z) \mid f(x, y, z) = o\}| \leq |\{(y, z) \mid f(o, y, z) = o\}| \]
and thus Equation (\ref{auction:eq:vuln}) becomes:
\begin{equation}
\label{auction:eq:vuln3}
\V = \frac{1}{m^3} \sum_o |\{(y, z) \mid f(o, y, z) = o\}| 
\end{equation}

The aim is thus now to compute the above sum, which can be written as the cardinality of the following set $S$:
\[ S = \{ (x, y, z) \mid f(x, y, z) = x \}\]
which can be written as the following disjoint union:
\[ S = \bigcup_x S^x \]
where for all $x$ we define:
\[ S^x = \{ (x, y, z) \mid f(x, y, z) = x \} \]

Let us partition each set $S^x$ into the following four subsets, that correspond to the different possible input orderings:
\begin{align*}
S_1^x &= S \cap \{ (x, y, z) \mid x > y \wedge x > z \} \\
S_2^x &= S \cap \{ (x, y, z) \mid y \geq x > z \} \\
S_3^x &= S \cap \{ (x, y, z) \mid z \geq x > y \} \\
S_4^x &= S \cap \{ (x, y, z) \mid x \leq y \wedge x \leq z \}
\end{align*}

We note that $S_1$ corresponds to the cases where $x_1$ is the largest of $x, y$ and $z$, $S_2$ and $S_3$ includes the cases where $x_1$ is the middle element and $S_4$ depicts the cases where $x_1$ is the smallest item. 

Let $x$ be in $\lr{1}{m}$. 
It is immediate to see that $\{ S_1^x, \dots, S_4^x \}$ forms a partition of $S^x$ and thus $|S^x| = |S_1^x|+\dots +|S_4^x|$, and furthermore $|S| = \sum_x |S^x|$. Let us thus focus on the cardinality of those four subsets. 

\ourparagraph{Case 1}
Let $y$ and $z$ be in $\lr{1}{m}$ such that:
\begin{align*}
x > y \wedge
x > z
\end{align*}
and let us study the membership of $(x, y, z)$ in $S_1^x$. 

We have:
\begin{align*}
f(x, y, z) = x \iff {}&
	        (x > 2y \wedge
	        x > 3z)
\\{}\vee{}&
	        (x > 3y \wedge
	        x > 2z)
\\ \iff {}&
	        (y < \frac{1}{2}x \wedge
	        z < \frac{1}{3}x)
\\{}\vee{}&
	        (y < \frac{1}{3}x \wedge
	        z < \frac{1}{2}x)
\end{align*}

In order to tally the number of different pairs $(y, z)$ that satisfy those conditions, it is helpful to rewrite those systems as the following disjunction:

\begin{align*}
f(x, y, z) = x \iff {}&
	        (y < \frac{1}{3}x \wedge
	        z < \frac{1}{3}x)
\\ {} \vee {}&
	        (y < \frac{1}{3}x \wedge
	        \frac{1}{3}x \leq z < \frac{1}{2}x)
\\ {} \vee {}&
	        (\frac{1}{3}x \leq y < \frac{1}{2}x \wedge
	        z < \frac{1}{3}x)
\end{align*}

The three disjuncts above are disjoint, and the last two disjuncts are symmetrical in $y$ and $z$. We can thus express the cardinality of $S_1^x$ as:
\begin{align}
|S_1^x| &= \ceil{ \frac{1}{3}x - 1 }^2 + 2\left( \ceil{\frac{1}{2}x} - 1 - \ceil{\frac{1}{3}x} + 1 \right)\left( \ceil{\frac{1}{3}x - 1} \right) \nonumber \\
&= \left( \ceil{\frac{1}{3}x - 1} \right) \left( 2\ceil{\frac{1}{2}x} - \ceil{\frac{1}{3}x} - 1 \right) \label{auction:eq:s1}
\end{align}

\ourparagraph{Case 2 and 3}
Let us now study the cardinality of $S_2^x$. 
Let $y$ and $z$ be in $\lr{1}{m}$ such that:
\begin{align*}
y \geq x > z
\end{align*}
and let us study the membership of $(x, y, z)$ in $S_1^x$. 
We have:

\begin{align*}
f(x, y, z) = x &\iff
	        (2x \geq y \wedge
	        2x > 3z)
\\ &\iff
	        (x \leq y \leq 2x \wedge
	        z < \frac{2}{3}x)
\end{align*}

We can thus express the cardinality of $S_2^x$ depending on which side of $n$, value $2x$ is. 
If $2x \leq m$, then:
\begin{align}
|S_2^x| &= (2x - x + 1) \left( \ceil{\frac{2}{3}x} - 1 \right) \nonumber \\
&= (x+1) \left( \ceil{\frac{2}{3}x} - 1 \right) \label{auction:eq:s2a}
\end{align}

If $2x > m$, then:
\begin{align}
|S_2^x| &= (m - x + 1) \left( \ceil{\frac{2}{3}x} - 1 \right) \label{auction:eq:s2b}
\end{align}

By symmetry on $y$ and $z$, we also have $|S_3^x| = |S_2^x|$. 

\ourparagraph{Case 4}
Let us now take $y$ and $z$ in $\lr{1}{m}$ such that:
\begin{align*}
x \leq y \wedge
x \leq z
\end{align*}
and let us study the membership of $(x, y, z)$ in $S_4^x$. 

We have:
\begin{align*}
f(x, y, z) = x \iff {}&
	        (3x \geq 2y \wedge
	        3x \geq z)
\\ {} \vee {}&
	        (3x \geq y \wedge
	        3x \geq 2z)
\\ \iff {} &
	        (x \leq y \leq \frac{3}{2}x \wedge
	        x \leq z \leq 3x)
\\ {} \vee {}&
	        (x \leq y \leq 3x \wedge
	        x \leq z \leq \frac{3}{2}x)
\end{align*}

Rewriting those conditions as disjoint cases, we get:
\begin{align*}
f(x, y, z) = x \iff {}&
	        (x \leq y \leq \frac{3}{2}x \wedge
	        x \leq z \leq \frac{3}{2}x)
\\ {} \vee {}&
	        (x \leq y \leq \frac{3}{2}x \wedge
	        \frac{3}{2}x < z \leq 3x)
\\ {} \vee {}&
	        (\frac{3}{2}x < y \leq 3x \wedge
	        x \leq z \leq \frac{3}{2}x)
\end{align*}

Let us treat those 3 disjoint disjunctions separately. 
Let $c_1^x, c_2^x$ and $c_3^x$ denote the number of different triples $(x, y, z)$ that satisfy the three above systems respectively, i.e. that:
\begin{align*}
c_1^x &= \left| \left\{ (x, y, z) \mid x \leq y \leq \frac{3}{2}x \wedge x \leq z \leq \frac{3}{2}x \right\} \right| \\
c_2^x &= \left| \left\{ (x, y, z) \mid x \leq y \leq \frac{3}{2}x \wedge \frac{3}{2}x < z \leq 3x \right\} \right| \\
c_3^x &= \left| \left\{ (x, y, z) \mid \frac{3}{2}x < y \leq 3x \wedge x \leq z \leq \frac{3}{2}x \right\} \right| \\
\end{align*}

Let us focus on $c_1^x$ first. 
If $\frac{3}{2}x \leq m$, then:
\begin{equation}
\label{auction:eq:c1a}
 c_1^x = \left( \floor{\frac{3}{2}x} - x + 1 \right)^2 
\end{equation}
Otherwise, if $\frac{3}{2}x > m$, then:
\begin{equation}
\label{auction:eq:c1b}
 c_1^x = (m - x + 1)^2 
 \end{equation}

Let us now focus on $c_2^x$. 
\begin{equation}
\label{auction:eq:c2}
c_2^x =
\begin{dcases*}
\left( 3x - \floor{\frac{3}{2}x} \right) \left( \floor{\frac{3}{2}x} - x  + 1 \right) 
	& if $3x\leq m$ \\
\left( m - \floor{\frac{3}{2}x} \right) \left( \floor{\frac{3}{2}x} - x  + 1 \right) 
	& if $\frac{3}{2}x\leq m \leq 3x$ \\
0 & if $\frac{3}{2}x > m$
\end{dcases*}
\end{equation}

And by symmetry on $y$ and $z$, we have $c_2^x = c_3^x$. 
Finally, we have $|S_4^x| = c_1^x + c_2^x + c_3^x$.

Let us recall that we wished to compute the cardinality of set $S$ in order to fulfil our original aim which was to compute input $x$'s vulnerability $\V$. 
We have now derived the cardinality of each subset $S_i^x$ for all $i$ in $\lr{1}{4}$ and for all $x$ in $\lr{1}{m}$. 
Moreover, we have intentionally expressed $S$ as a disjoint union so that: 
\begin{equation}
\label{auction:eq:cardSsum}
|S| = \sum_x |S_1^x|+\dots+|S_4^x| 
\end{equation}

The expression that we can derive for $|S|$ in this way involves sums with ceilings and floorings and would thus not immediately lead to a closed-form formula that can be computed in constant time. 
Indeed, combining Equations (\ref{auction:eq:s1}), (\ref{auction:eq:s2a}), (\ref{auction:eq:s2b}), (\ref{auction:eq:c1a}), (\ref{auction:eq:c1b}) and (\ref{auction:eq:c2}) provides us with a closed-form expression for $|S_1^x|+\dots+|S_4^x|$ for any fixed value of $x$ in $\lr{1}{m}$. 
Equation (\ref{auction:eq:cardSsum}) then allows us to compute the cardinality $|S|$ by summing those $m$ expressions, which allows us to compute $|S|$ in $\cO(m)$ time. 


However, we recall that for small input spaces, we already have a combinatorial way of computing vulnerability $\V$, and that our major problem is to scale our analyses to \emph{large} input spaces, which is specifically where the combinatorial method fails to scale. 

For that reason, in the remainder of this report, we will aim at deriving a closed-form formula for the asymptotic behaviour of $|S|$ for large values of input size $m$. 

\subsection{Asymptotic Behaviour of the Input Vulnerability}

In order to do so, we will study the asymptotic behaviour, when $m$ tends towards infinity, of $|S_i^x|$ and $\sum_x |S_i^x|$ for all $i$ in $\lr{1}{4}$ in order to be able to compute that of $|S|$ and thus of $\V$. 
In this section, the asymptotic behaviour of the cardinality of a set, say $S$, will refer to the asymptotic behaviour of $|S|$ when expressed as a function of $m$, when $m$ tends towards infinity. 

Let us consider again the expression of $|S_1^x|$ obtained in Equation (\ref{auction:eq:s1}) and let us study its asymptotic behaviour when $m$ is large. We note that $x$ is an integer ranged in $\lr{1}{m}$ and is thus a $\cO(m)$. 
By simplifying the ceiling in the first factor, we have:
\begin{align*}
\ceil{\frac{1}{3}x -1} &= \frac{1}{3}x + \cO(1) 
\end{align*}

\ourparagraph{Case 1}
We now recall the entire expression of $|S_1^x|$ obtained in Equation (\ref{auction:eq:s1}) and proceed with a similar reasoning:
\begin{align*}
|S_1^x| &= \left( \ceil{\frac{1}{3}x - 1} \right) \left( 2\ceil{\frac{1}{2}x} - \ceil{\frac{1}{3}x} - 1 \right) \\
&= \left( \frac{1}{3}x + \cO(1)  \right) \left( x - \frac{1}{3}x + \cO(1) \right) \\
&= \frac{2}{3}x^2 + x \cdot \cO(1) \\
&= \frac{2}{3}x^2 + \cO(m)
\end{align*}

We thus obtain:
\begin{align*}
\sum_x |S_1^x| &= \sum_{x=1}^m \left( \frac{2}{3} x^2  + \cO(m) \right) \\
&= \frac{2}{3} \cdot  \frac{m(m+1)(2m+1)}{6} + \cO(m^2) \\
&= \frac{1}{3^2} m^3 + \cO(m^2)
\end{align*}

\ourparagraph{Case 2 and 3}
Let us now study the case of $|S_2^x|$. 
Based on Equations (\ref{auction:eq:s2a}), then if $2x \leq m$, we have:
\begin{align*}
|S_2^x| 
&= (x+1) \left( \ceil{\frac{2}{3}x} - 1 \right) \\
&= \frac{2}{3}x^2 + \cO(m)
\end{align*}

Similarly, if $2x > m$, then Equation (\ref{auction:eq:s2b}) becomes:
\begin{align*}
|S_2^x| &= (m - x + 1) \left( \ceil{\frac{2}{3}x} - 1 \right) \\
&= \frac{2}{3}mx - \frac{2}{3}x^2 + \cO(m)
\end{align*}

And thus:
\begin{align*}
\sum_{x=1}^m |S_2^x|
&= \sum_{x=1}^{\frac{m}{2}} \left( \frac{2}{3}x^2 + \cO(m) \right) + \sum_{x=\frac{m}{2}}^m \left( \frac{2}{3}mx - \frac{2}{3}x^2 + \cO(m) \right) + \cO(m^2)
\end{align*}
where we note again that the bounds of the indices in the sums are allowed to differ by 1 since the difference is compensated in the $\cO(n^2)$ term. 

Let us compute each term separately. Using the results recalled in Lemma \ref{auction:lemma:sums}, we have:
\begin{align*}
\sum_{x=1}^{\frac{m}{2}} \left( \frac{2}{3}x^2 + \cO(m) \right)
&= \frac{2}{3} \cdot \frac{1}{3} \left( \frac{m}{2} \right)^3 + \cO(m^2) \\
&= \frac{1}{2^2 \cdot 3^2} m^3 + \cO(m^2)
\end{align*}

Similarly, we note that $\sum_{x=\frac{m}{2}}^m \cO(m) = \cO(m^2)$. Moreover, we have:
\begin{align*}
\sum_{x=\frac{m}{2}}^m \left( \frac{2}{3}mx - \frac{2}{3}x^2 \right)
&= \frac{2}{3}m \sum_{x=\frac{m}{2}}^m x - \frac{2}{3} \sum_{x=\frac{m}{2}}^m x^2 + \cO(m^2) \\
&= \frac{2m}{2\cdot 3} \left(m^2 - \left( \frac{m}{2} \right)^2\right) - \frac{2}{3^2} \left( m^3 - \left( \frac{m}{2} \right)^3 \right) + \cO(m^2) \\
&= \frac{2m}{2\cdot 3} \cdot \frac{3m^2}{2^2} - \frac{2}{3^2} \cdot \frac{7m^3}{2^3} + \cO(m^2) \\
&= \frac{m^3}{2^2} - \frac{7m^3}{2^2 \cdot 3^2} + \cO(m^2) \\
&= \frac{1}{2\cdot 3^2} m^3 + \cO(m^2) 
\end{align*}

Combining the previous two terms, we get:
\begin{align*}
\sum_{x=1}^m |S_2^x| 
&= \frac{1}{2^2 \cdot 3^2} m^3 + \frac{1}{2\cdot 3^2} m^3 + \cO(m^2) \\
&= \frac{1}{2^2 \cdot 3} m^3 + \cO(m^2)
\end{align*}

By symmetry, we immediately get:
\begin{align*}
\sum_{x=1}^m ( |S_2^x| + |S_3^x| )
&= \frac{1}{2 \cdot 3} m^3 + \cO(m^2)
\end{align*}

\ourparagraph{Case 4}
Let us finally study the asymptotic behaviour of $\sum_x |S_4^x| = \sum_x (c_1^x + c_2^x + c_3^x)$. 
If $\frac{3}{2}x \leq m$, we know from Equation (\ref{auction:eq:c1a}) that:
\begin{align*}
c_1^x &= \left( \floor{\frac{3}{2}x} - x + 1 \right)^2 \\
&= \frac{x^2}{2^2} + \cO(m)
\end{align*}

Moreover, if $\frac{3}{2}x > m$, then from Equation (\ref{auction:eq:c1b}):
\begin{align*}
c_1^x &= (m - x + 1)^2 \\
&= m^2 - 2mx + x^2 + \cO(m)
\end{align*}

Thus:
\begin{align*}
\sum_{x=1}^m c_1^x 
&= \sum_{x=1}^{\frac{2}{3}m} \frac{x^2}{2^2} + \sum_{x=\frac{2}{3}m}^m (m^2 - 2mx + x^2) + \cO(m^2) \\
&= \frac{1}{2^2\cdot 3} \cdot \frac{2^3 m^3}{3^3} + \frac{m^3}{3} - \frac{2\cdot 5m^3}{2\cdot 3^2} + \frac{19m^3}{3^4} + \cO(m^2) \\
&= \frac{m^3}{3^3} + \cO(m^2)
\end{align*}

Similarly, by Equation (\ref{auction:eq:c2}), we have:
\[
c_2^x =
\begin{dcases*}
\frac{3}{4}x^2 + \cO(m)
	& if $3x\leq m$ \\
\frac{xm}{2} - \frac{3}{4}x^2 + \cO(m)
	& if $\frac{3}{2}x\leq m \leq 3x$ \\
0 & if $\frac{3}{2}x > m$
\end{dcases*}
\]
%

Thus:
\begin{align*}
\sum_{x=1}^m c_2^x
&=
\sum_{x=1}^{\frac{m}{3}} \frac{3}{4}x^2
+
\sum_{x=\frac{m}{3}}^{\frac{2m}{3}} \left( \frac{xm}{2} - \frac{3}{4}x^2 \right)
+ \cO(m^2) \\
&= \frac{3m^3}{4\cdot 3 \cdot 3^3} + \frac{(2^2-1)m^3}{2\cdot 2\cdot 3^2} - \frac{7\cdot 3m^3}{4\cdot 3 \cdot 3^3} + \cO(m^2) \\
&= \frac{m^3}{2^2 \cdot 3^3} + \cO(m^2)
\end{align*}

And by symmetry, we have:
\begin{align*}
\sum_{x=1}^m (c_2^x + c_3^x)
&=
\frac{m^3}{2 \cdot 3^3} + \cO(m^2)
\end{align*}

And thus:
\begin{align*}
\sum_{x=1}^m |S_4^x|
&=
\frac{m^3}{3^3} + \frac{m^3}{2 \cdot 3^3} + \cO(m^2) \\
&= \frac{m^3}{2 \cdot 3^2} + \cO(m^2)
\end{align*}

\ourparagraph{Conclusion}
And finally:
\begin{align}
|S| &= \sum_{i=1}^4 |S_i| \nonumber \\
&= \frac{1}{3^2} m^3 + \frac{1}{2 \cdot 3} m^3 + \frac{m^3}{2 \cdot 3^2} + \cO(m^2) \nonumber\\
&= \frac{1}{3}m^3 + \cO(m^2) \label{auction:eq:cardS}
\end{align}

Given this analysis, we can now formulate our main result for 3-party auctions. 

\begin{theorem}
\label{auction:thm:h3}
We have:
\[
\HH(X \mid O) = \log 3 + \cO \left(\frac{1}{m}\right)
\]
\end{theorem}

\begin{proof}
Equation (\ref{auction:eq:vuln3}) and (\ref{auction:eq:cardS}) yield:
\[ \V = \frac{1}{3} + \cO\left(\frac{1}{m}\right) \]
and manipulating asymptotic developments yields:
\begin{align*}
\HH
&= -\log \left( \frac{1}{3} \left( 1 + \cO\left(\frac{1}{m}\right) \right)\right) \\
&= \log 3 -\log \left( 1 + \cO\left(\frac{1}{m}\right) \right) \\
&= \log 3 + \cO\left(\frac{1}{m}\right)
\end{align*}%
\end{proof}

It follows that $\HH(X \mid O)$ converges when $m$ tends toward infinity and its limit is stated in the following result. 

\begin{corollary}
\label{auction:coro:3}
When the input size $m$ tends towards infinity, the value of $\HH(X \mid O)$ converges and:
\[ \lim_{m \to \infty} \HH(X \mid O) = \log 3 \]
\end{corollary}

\begin{proof}
This is an immediate consequence of Theorem \ref{auction:thm:h3}. 
\end{proof}

We note that a more precise approximation of $\V$ and $\HH(X \mid O)$ can be obtained by bounding the cardinality of the subsets involved in the computation of $|S|$ with upper and lower bounds. In particular, each flooring and ceiling can be approximated with a range of width 1, and so more involved computation would lead to a more precise result. 

The results obtained so far led us to conjecture on the asymptotic vulnerability of a targeted input in a general $n$-party auction.

\section{Conjecture for Multi-Party Auctions}
\label{auction:sec:multi}

Let $n$ be a positive integer and let us consider an $n$-party auction. 
Let $m$ be a positive integer that represents the input size. 
Let us consider $n$ inputs $x_1, \dots, x_n$ and we recall that the computation of the output of the auction is modelled by function $f$. 

In light of Corollaries \ref{auction:coro:2} and \ref{auction:coro:3}, we formulate the following conjecture. 

\begin{conjecture}
\label{auction:conj:limit}
Let $\HH(X \mid O)$ represent the conditional entropy of one of the inputs, $X$, given the output of the auction $O$. 
Then $\HH(X \mid O)$ converges when $m$ tends towards infinity and:
\[ \lim_{m \to \infty} \HH(X \mid O) = \log n \]
\end{conjecture}

We note that this conjecture is also verified when $n = 1$ since in the presence of a single party, we have $\HH(X \mid O) = 0$. 
The results derived in Corollaries \ref{auction:coro:2} and \ref{auction:coro:3} are also consistent with this conjecture. 
In this section, we provide more evidence which supports our intuition.

Let us first generalise a notion introduced in Equation (\ref{auction:eq:vuln3}) that will enable us to express the vulnerability more conveniently. 

\begin{definition}
We define function $c_n\colon \mathbb{N}^* \to \mathbb{N}$ for all positive integer $m$ as:
\[ c_n(m) = |\{ (x_1, \dots, x_n) \in \lr{1}{m}^n \mid f(x_1, \dots, x_n) = x_1 \}| \]
\end{definition}

As discussed in Section \ref{auction:sec:three}, we can easily prove that $\V = \frac{c_n(m)}{m^n}$. 
The difficulty resides in computing $c_n(m)$ for large values of $m$. 

Let us first formulate a second conjecture on the shape of function $c_n$. This will help us to reason about its asymptotic behaviour more precisely. 

\begin{conjecture}
\label{auction:conj:shape}
There exists a positive rational number $a_n$ in $\mathbb{Q}$ such that:
\[ c_n(m) = a_n \cdot m^n + \cO(m^{n-1}) \]
\end{conjecture}

We can see in Section \ref{auction:sec:three} from the way that $c_n(m)$ is constructed that it will comply with the shape aforementioned. 

Conjecture \ref{auction:conj:limit} is now equivalent to the fact that $a_n = \frac{1}{n}$, for which we will show some empirical supportive evidence. 
We emphasise the fact that the following reasoning is empirical and does not constitute a proof. 

As Conjecture \ref{auction:conj:shape} suggests that function $c_n$ behaves like a polynomial of degree $n$ for large values of $m$, we decided to try and interpolate function $c$ with a polynomial of degree at most $n$ via its first 30 values $c_n(1), \dots, c_n(30)$ which we computed empirically. 
Let us plot the first 30 values of $c_n$ that we computed for $n=2, \dots, 5$. 
Based on those 20 values, we performed a polynomial regression with a polynomial of degree $n$ respectively, and with a least squares method. 
The respective interpolating polynomials $P_2, \dots, P_5$ that we empirically obtain are displayed below:
\begin{align*}
P_2 &=  0.5 x^2 + 0.5 x \\
P_3 &= 0.3334 x^3 + 0.6649 x^2 + 0.3577 x - 0.1987 \\
P_4 &= 0.2499 x^4 + 0.7671 x^3 + 0.3552 x^2 + 1.267 x - 2.231 \\
P_5 &= 0.1995 x^5 + 0.8503 x^4 - 0.6201 x^3 + 14.77 x^2 - 62.27 x + 64.15
\end{align*}
where coefficients have been reported with 4 significant figures. 

By looking at the coefficients of highest degree in the above polynomials, we can see that the empirical values thus obtained for $a_n$ are indeed very close to $\frac{1}{n}$, which supports our Conjecture \ref{auction:conj:limit}.

\begin{figure}
\centering

\subfigure[Case $n=2$. ]{
\begin{tikzpicture}[scale=1]
	\begin{axis}[
	 ymin=0,
	 xmin=0,
	  xlabel=$m$,
	  ylabel=prevalence,
	  width=8cm,
	  height=9cm,
		scaled y ticks=base 10:-2,
	  legend pos=north west]
	\addplot[mark size=1pt, only marks, blue] table [y=y, x=x]{figuresAuction/2.dat};
	\addplot[red] table [y=y, x=x, mark=none]{figuresAuction/2b.dat};
	\legend{$c_n(m)$, $P_n(m)$}
\end{axis}
\end{tikzpicture}
\label{auction:fig:n2}
%
%

}\hfill
\subfigure[Case $n=3$. ]{
\begin{tikzpicture}[scale=1]
	\begin{axis}[
	 ymin=0,
	 xmin=0,
	  width=8cm,
	  height=9cm,
	  xlabel=$m$,
	  ylabel=prevalence,
	  scaled y ticks=base 10:-3,
	  legend pos=north west]
	\addplot[mark size=1pt, only marks, blue] table [y=y, x=x]{figuresAuction/3.dat};
	\addplot[red] table [y=y, x=x, mark=none]{figuresAuction/3b.dat};
	\legend{$c_n(m)$, $P_n(m)$}
\end{axis}
\end{tikzpicture}
\label{auction:fig:n3}
%
%
}
\par\bigskip
\par\bigskip
\subfigure[Case $n=4$. ]{
\begin{tikzpicture}[scale=1]
	\begin{axis}[
	 ymin=0,
	 xmin=0,
	  width=8cm,
	  height=9cm,
	  xlabel=$m$,
	  ylabel=prevalence,
	  scaled y ticks=base 10:-4,
	  legend pos=north west]
	\addplot[mark size=1pt, only marks, blue] table [y=y, x=x]{figuresAuction/4.dat};
	\addplot[red] table [y=y, x=x, mark=none]{figuresAuction/4b.dat};
	\legend{$c_n(m)$, $P_n(m)$}
\end{axis}
\end{tikzpicture}
\label{auction:fig:n4}
%
%
}\hfill
\subfigure[Case $n=5$.]{
\begin{tikzpicture}[scale=1]
	\begin{axis}[
	 ymin=0,
	 xmin=0,
	  width=8cm,
	  height=9cm,
	  xlabel=$m$,
	  ylabel=prevalence,
	  legend pos=north west]
	\addplot[mark size=1pt, only marks, blue] table [y=y, x=x]{figuresAuction/5.dat};
	\addplot[red] table [y=y, x=x, mark=none]{figuresAuction/5b.dat};
	\legend{$c_n(m)$, $P_n(m)$}
\end{axis}
\end{tikzpicture}
\label{auction:fig:n5}
}
\caption{Polynomial interpolation of empirical values of $c_n$ with theoretical polynomials $P_n$ of degree at most $n$. The first 30 values of $c_n$ were computed empirically, those values were used to perform a polynomial regression with a least squares method. }
\end{figure}
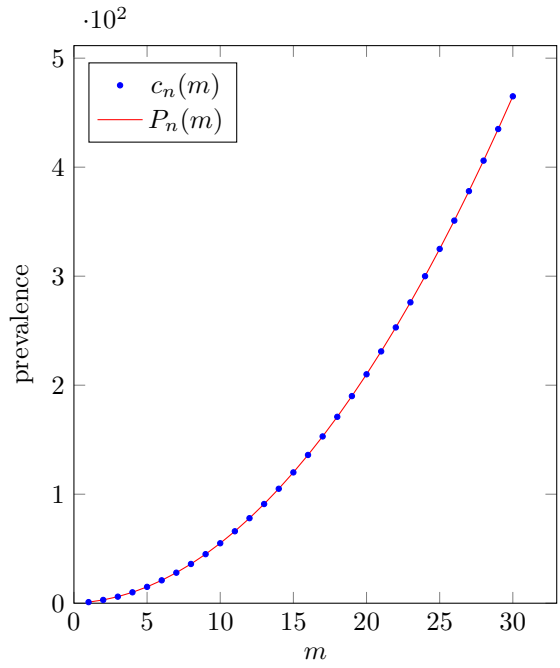
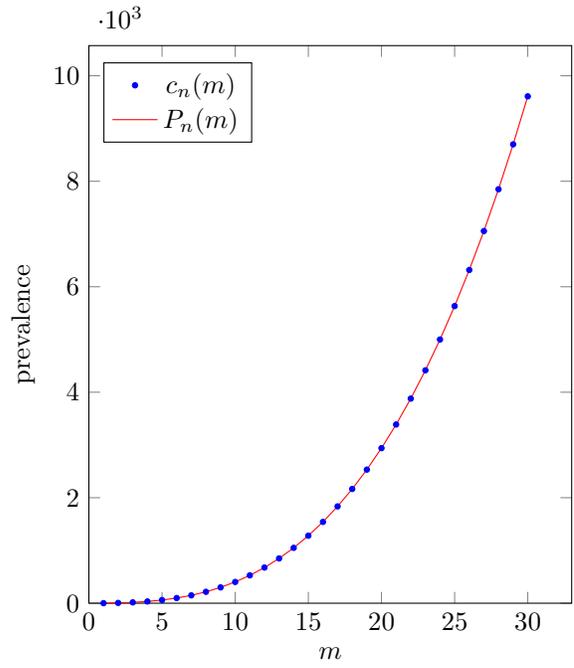
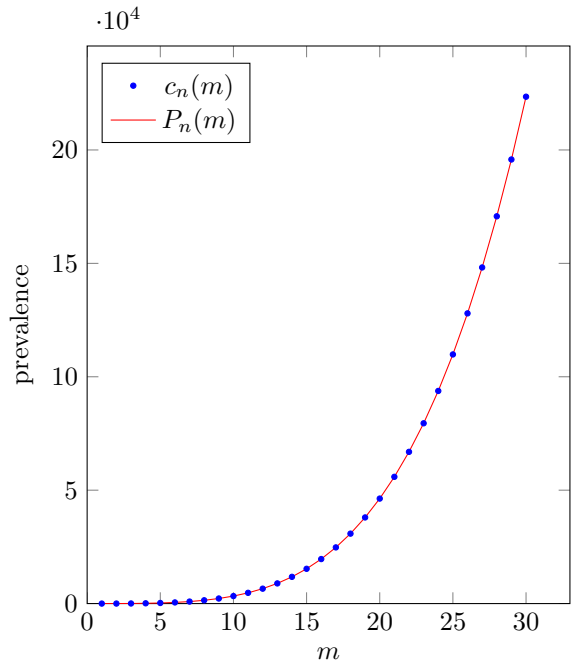
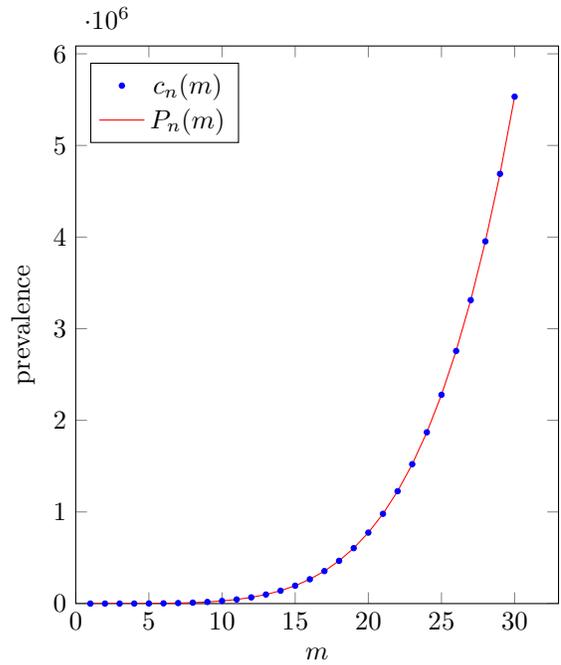

\section{Discussion}
\label{sec:discu}

The results obtained and conjectured in this paper suggest that the computational power that an attacker may have should somehow play a role in our approach. 
Similarly to the differences that have been formally defined between information-theoretic security and computational security in a cryptographic protocol, we may also identify similar nuances in privacy. 

For example, let us assume that our Conjecture \ref{auction:conj:limit} is correct and that the vulnerability of one input in an $n$-party auction converges towards $\frac{1}{n}$ when the input size $m$ tends to infinity.
This would be a poor privacy guarantee, since the prior vulnerability of one input is $\frac{1}{m}$, and $m$ is much larger than $n$.

However, this $\frac{1}{n}$ limit would be a \emph{theoretical} value: this means that in theory, an attacker learning the output of an auction has on average a $\frac{1}{n}$ probability of guessing one input in one try if he selects the best guess.
But in reality, an attacker with limited computational power might not be able to select the best guess that would offer him a $\frac{1}{n}$ probability of guessing the secret.

Our empirical method for computing $c_n(m)$~--~that does not scale to large input spaces~--~does provide, along with the exact value of $c_n(m)$, the best guessing strategy that achieves the expected vulnerability as it explicitly chooses the best guess.

However, in the 3-party auction, we solely proved that for a large value of $m$, the input's vulnerability would be close to $\frac{1}{3}$.
In particular, we were unable to provide a way for an attacker to select the best guessing strategy given the auction's output.

One might thus be interested in considering the potential difference that may exist between the theoretical vulnerability and the computational vulnerability of a secret.

\section{Conclusion}
\label{sec:conclu}

Digital goods auctions are one real world use case that can benefit from the security guarantees that Secure Multi-Party Computation has to offer. 
One of the main advantages of using SMC is to protect the confidentiality of the participants' bids. 
As in every application of SMC, private inputs in auctions are subject to acceptable leakage. 
Although general, combinatorial privacy analyses are able to quantify this leakage for small input spaces, they fail to scale to large input spaces. 
In this paper, we derived methods for quantifying the acceptable leakage that scales to arbitrarily large input spaces in the particular case of digital goods auctions. 
We first derived a closed-form formula for the posterior entropy of a targeted input in two-party auctions. 
We then focused on studying the asymptotic behaviour of this posterior entropy in three-party auctions. 
This enabled us to formulate a conjecture on the asymptotic behaviour of this acceptable leakage in general $n$-party auctions with large input spaces, which we further supported with empirical observations. 

\bibliographystyle{plain} 
\bibliography{./references}

\end{document}